\documentclass[11pt]{article}
\usepackage[margin=1in]{geometry}
\usepackage{amsmath, amssymb, amsthm}
\usepackage{charter}
\usepackage{euler}
\usepackage[round,authoryear]{natbib}
\usepackage[
  colorlinks=true,
  linkcolor=blue!60!black,
  citecolor=green!45!black,
  urlcolor=magenta!60!black
]{hyperref}
\usepackage[capitalize]{cleveref}
\let\cite\citep
\usepackage{xcolor}
\usepackage{enumitem}
\usepackage{booktabs}
\usepackage{tikz}
\usetikzlibrary{calc, arrows.meta, decorations.pathmorphing}
\usepackage{algorithm}
\usepackage[noend]{algpseudocode}

\newcommand{\R}{\mathbb{R}}

\newcommand{\dGH}{d_{\mathrm{GH}}}
\newcommand{\dGHbij}{d_{\mathrm{GH}}^{\,\mathrm{bij}}}
\newcommand{\dHiso}{d_{H,\mathrm{iso}}}
\newcommand{\dBiso}{d_{B,\mathrm{iso}}}
\newcommand{\dH}{d_H}
\newcommand{\Dist}{\mathit{Dist}}

\newcommand{\Iso}{\mathrm{Iso}}
\newcommand{\diam}{\mathrm{diam}}
\newcommand{\dist}{\mathrm{dist}}

\usepackage{thmtools}

\declaretheorem[name=Theorem,numberwithin=section]{theorem}
\declaretheorem[name=Lemma,sibling=theorem]{lemma}
\declaretheorem[name=Proposition,sibling=theorem]{proposition}
\declaretheorem[name=Corollary,sibling=theorem]{corollary}

\declaretheorem[name=Definition,sibling=theorem,style=definition]{definition}
\declaretheorem[name=Remark,sibling=theorem,style=definition]{remark}

\crefname{lemma}{Lemma}{Lemmas}
\crefname{proposition}{Proposition}{Propositions}
\crefname{corollary}{Corollary}{Corollaries}
\crefname{conjecture}{Conjecture}{Conjectures}
\crefname{hypothesis}{Hypothesis}{Hypotheses}
\crefname{problem}{Problem}{Problems}
\crefname{definition}{Definition}{Definitions}
\crefname{remark}{Remark}{Remarks}

\title{A constant-factor approximation of the\\
Gromov--Hausdorff distance in the plane}

\author{Sushovan Majhi%
\thanks{Data Science Program, George Washington University, Washington, DC, USA.
Email: \texttt{s.majhi@gwu.edu}.}}
\date{}

\begin{document}
\maketitle

\begin{abstract}
We give the first polynomial-time constant-factor approximation of the Gromov--Hausdorff distance
$\dGH$ between finite point sets in the Euclidean plane; in fixed Euclidean dimension such an
approximation was previously known only on the line \citep{MVW2024}. Global alignment cannot
succeed: the classical dimension drop defeats alignment by isometries, a multiplicity gap defeats
alignment by bijections, and a reflection barrier defeats sorting under any single global
reflection pattern.
The algorithm is therefore local. Guessing the images of one diameter pair pins every point's
longitudinal coordinate to within $O(\dGH)$. Heights are read in windows whose height spread is at
most a fixed multiple of their length, where a chain argument makes every compatible match local in the plane. One
reflection sign per window is then chosen by 2-SAT; at the right frame and guess, any solution
yields a correspondence of distortion $O(\dGH)$. For the
bijective relative $\dGHbij$, half the least additive distortion over bijections, the same scheme reduces
the planar problem to a matching question that we leave open.
\end{abstract}



\smallskip
\noindent\textbf{ACM CCS:} Theory of computation $\rightarrow$
Approximation algorithms analysis; Mathematics of computing
$\rightarrow$ Combinatorial algorithms.

\section{Introduction}
\label{sec:intro}

How nearly isometric are two finite point sets $X, Y \subset \R^d$? This paper answers the
algorithmic form of that question in the plane: we give the first polynomial-time constant-factor
approximation of the \emph{Gromov--Hausdorff distance} $\dGH$ between finite point sets in $\R^2$
(\Cref{thm:main}). In fixed Euclidean dimension such approximations were previously known only on
the line \citep{MVW2024}.

The distance $\dGH$ \cite{Gromov2007,BBI2001} is the canonical measure of near-isometry for
arbitrary compact metric
spaces, vanishing if and only if
they are isometric; it underlies shape comparison and topological data
analysis \cite{MemoliSapiro2005,Memoli2012}. It has an equivalent form through
\emph{correspondences}---relations $R \subseteq X \times Y$ projecting surjectively onto both factors:
\begin{equation}
\label{eq:dGH}
\dGH(X, Y) = \tfrac{1}{2} \inf_R \, \Dist(R), \qquad\text{where}\quad
\Dist(R) := \!\!\sup_{(x,y),(x',y') \in R}\!\!
\bigl| d_X(x, x') - d_Y(y, y') \bigr|.
\end{equation}
Throughout the paper $X$ and $Y$ are finite nonempty sets of Euclidean points, with $d_X, d_Y$ the standard
Euclidean distance restricted to them. When $|X| = |Y|$, restricting correspondences to bijections $\sigma : X \to Y$ gives the
\emph{bijective} (or \emph{bottleneck}) Gromov--Hausdorff distance~\citep{Schmiedl2017}
\begin{equation}
\label{eq:dGHbij}
\dGHbij(X, Y) := \tfrac12 \min_{\sigma : X \to Y \text{ bijection}} \Dist(\sigma),
\qquad\text{where}\quad
\Dist(\sigma) := \max_{i, j} \bigl| d_X(i, j) - d_Y(\sigma i, \sigma j) \bigr| ;
\end{equation}
thus $2\,\dGHbij$ is the minimum \emph{additive distortion}, the largest change a bijection makes to
a pairwise distance, and every bijection being a correspondence gives $\dGH \le \dGHbij$. Naive
evaluation of either is exponential in $|X| + |Y|$, so the algorithmic question is one of
approximation.

The additive-distortion question itself goes
back two decades: \citet{HallPapadimitriou2005} gave a polynomial-time $2$-approximation on the
line. They remark, without giving a proof, that a modification of the construction of
\citet{PapadimitriouSafra2005} makes approximation within $3$ NP-hard in dimension three.
The older \emph{multiplicative} distortion is known to be exactly computable on the line when it is below
$3 + 2\sqrt2$ \citep{KRS2009} and hard to approximate there when large
\citep{HallPapadimitriou2005}; it is NP-hard within $3$ in dimension three
\citep{PapadimitriouSafra2005}, the plane a noted open problem \citep{KRS2009}. On general metric spaces both distances are hard: $\dGHbij$ carries strong inapproximability
\citep{Schmiedl2017}, and approximating $\dGH$ itself within a factor of $3$ is NP-hard already for
geodesic metrics on trees, where only super-constant factors are achieved \citep{AFNSW2018}. The
contribution here is to exploit Euclidean structure in the plane to escape that worst-case wall.

Two extrinsic relatives of $\dGH$ recur throughout. For nonempty compact $A, B \subset \R^d$, the
\emph{Hausdorff distance} is
\begin{equation}
\label{eq:dH}
\dH(A, B) := \max\Bigl\{\, \sup_{a \in A} \inf_{b \in B} \|a - b\|,
\ \ \sup_{b \in B} \inf_{a \in A} \|a - b\| \,\Bigr\},
\end{equation}
and the \emph{isometric Hausdorff distance} minimizes it over the group $\Iso(\R^d)$ of Euclidean
isometries (rotations, reflections, and translations),
\begin{equation}
\label{eq:dHiso}
\dHiso(X, Y) := \inf_{T \in \Iso(\R^d)} \dH\bigl(T(X), Y\bigr).
\end{equation}
Each isometry $T$ induces a nearest-neighbor correspondence of distortion at most $2\,\dH(T(X), Y)$,
so $\dGH \le \dHiso$ always, while the reverse can fail by an unbounded factor
(\citealp[\S3.1]{Memoli2008}, after \citealp[2.12]{ATV2001}; see \Cref{thm:counter-r2});
$\dHiso$ itself is computable exactly in polynomial time in the plane \cite{ChewGoodrich1997}.

In fixed Euclidean dimension a constant-factor approximation of $\dGH$ was known only at $d = 1$, where
\citet{MVW2024} proved the stability bound $\dHiso(X, Y) \le \tfrac54\,\dGH(X, Y)$ which, with Rote's
$O(n\log n)$ algorithm for the Hausdorff distance under translation on the line \cite{Rote1991}, run
on $Y$ and on its mirror image, gives a polynomial-time
$\tfrac54$-approximation of $\dGH$ in $\R^1$. Two questions stand behind that result. \textbf{(Q1)} Does the
stability bound $\dHiso \le C\,\dGH$ extend to any $\R^d$ with $d \ge 2$? It does not, as
\citet{MVW2024} already record as known (see the prior work below and \Cref{thm:counter-r2}). \textbf{(Q2)} Is there a
polynomial-time constant-factor approximation of $\dGH$---or its bijective relative $\dGHbij$---in any
Euclidean dimension $\ge 2$? For $\dGH$, \citet{MVW2024} leave this open.

\paragraph{Prior work.}
Near-isometries of Euclidean sets were studied by \citet{ATV2001}. They showed that a map of a
bounded set in $\R^n$ distorting distances by at most $\varepsilon \le \diam$ lies within
$c_n\sqrt{\varepsilon\,\diam}$ of an isometry, with $c_n$ depending only on the dimension
\citep[Theorem~2.2]{ATV2001}. They also showed that the square root is necessary already for three
points: lifting the midpoint of $\{0, e_1/2, e_1\}$ by $h$ gives an $h^2$-nearisometry at distance
at least $h/2$ from every isometry \citep[2.12]{ATV2001}. On sets of positive thickness the error
becomes linear \citep[Theorem~3.3]{ATV2001}. Building on these results, \citet{Memoli2008} stated
the bound $\dHiso \le c'_d\sqrt{M\,\dGH}$ with $M := \max(\diam X, \diam Y)$, its proof deferred, and
adapted the three-point example to a planar family on which $\dHiso/\dGH$ is unbounded
\citep[Theorem~2 and \S3.1]{Memoli2008}. This answers Q1 negatively in every $\R^d$, $d \ge 2$.

\paragraph{Our contribution.}
We settle Q2 affirmatively for $\dGH$ in $\R^2$. A polynomial-time algorithm returns, for all finite
$X, Y \subset \R^2$, a correspondence whose distortion is within an absolute constant factor of
optimal (\Cref{thm:main}). The obstructions of \Cref{sec:dimension-drop} rule out the natural global
alignments: by isometry (the classical dimension drop), by bijection (a multiplicity gap), and by a
sorting under a single reflection pattern (a reflection barrier). So the algorithm works locally,
in three steps.
\begin{enumerate}[leftmargin=*]
\item \emph{Frame.} Guessing the images of one diameter pair of $X$ pins every point's
  longitudinal coordinate to within $O(\dGH)$ (\Cref{lem:frame}). Only the transverse coordinate,
  the \emph{height}, stays free, and it is free exactly to the extent of the dimension drop.
\item \emph{Windows.} Around each point, heights are read in a longitudinal window whose height
  spread is at most a fixed multiple of its length. Inside such a window an optimal correspondence acts on heights as a
  translation, up to one reflection sign (\Cref{lem:window}). A chain argument then makes every
  match compatible with the window local in the plane (\Cref{lem:locality,fig:chain}). A wrong height offset
  would force a climb, point by point, out through the window's edge, exceeding its spread.
\item \emph{Signs.} The remaining freedom is one reflection sign per window, and signs must be
  allowed to differ between windows (\Cref{rem:per-node-signs}). Pairwise distortion constraints
  on the signs form a 2-SAT instance, which the true signs satisfy (\Cref{lem:truth}). At the
  right frame and guess, any solution, with any choice of admissible targets, is a correspondence
  of distortion $O(\dGH)$ (\Cref{thm:main}).
\end{enumerate}
No step distinguishes near-collinear inputs (\Cref{sec:dimension-drop}) from others, and no step requires a matching, so sets of
different sizes need no separate treatment.

\paragraph{Main result.}
Let $X, Y \subset \R^2$ be finite, $n = |X| + |Y|$, $D = \diam X$, and, for $|X|, |Y| \ge 2$, $\delta_{\min}$ the least
positive value of $|d_X(a,b) - d_Y(a',b')|$ over $a, b \in X$ and $a', b' \in Y$. For every
$\varepsilon \in (0, 1]$, with $O(\varepsilon^{-1} n^6 \log(2 + D/\delta_{\min}))$ arithmetic operations,
\Cref{alg:main} computes a correspondence $R$ with
\[
\dGH(X,Y) \;\le\; \tfrac12\Dist(R) \;\le\; (1 + \varepsilon)\,K_0\,\dGH(X,Y)
\]
for an absolute constant $K_0 < 356$ (\Cref{thm:main}). The output is a witness: the correspondence
realizes the value it reports.

\paragraph{The bijective distance.}
For sets of equal size, $2\,\dGHbij$ is the least additive distortion of a bijection, the quantity
\citet{HallPapadimitriou2005} approximate on the line. The frame, windows and signs apply to
bijections unchanged. The last step, however, then needs a perfect matching compatible with the
chosen signs. We do not know that one always exists, and leave the planar bijective problem open
(\Cref{sec:discussion}).

\paragraph{Organization.}
\Cref{sec:dimension-drop} recalls the classical dimension drop and proves the multiplicity,
combined and reflection barriers. \Cref{sec:frame} sets up the frames and the windows.
\Cref{sec:locality} proves locality and sets up the signs. \Cref{sec:main} gives the algorithm, the
main theorem, the constants, and a lower bound for the algorithm. \Cref{sec:discussion} lists the
open questions.

\section{Obstructions to global alignment}
\label{sec:dimension-drop}

The $\R^1$ algorithm of~\citet{MVW2024} approximates $\dGH$ through $\dHiso$; the analogous
extrinsic route to the bijective distance $\dGHbij$ (half the least distortion over bijections) is
\emph{bottleneck matching} under isometry, defined for $|X| = |Y|$ by
\[
\dBiso(X, Y) := \inf_{T \in \Iso(\R^d)}\ \min_{\sigma : X \to Y\ \text{bijection}}\
\max_{x \in X}\ \|T(x) - \sigma(x)\| ,
\]
which dominates $\dGHbij$ exactly as $\dHiso$ dominates $\dGH$ ($\dGHbij \le \dBiso$ by the triangle
inequality). The \emph{dimension drop} (\Cref{thm:counter-r2,cor:counter-bij}) defeats both
extrinsic proxies at once. One three-point family---a collinear triple of diameter $D$ versus its
midpoint lifted by $2\delta \ll D$, intrinsically within $O(\delta^2/D)$ yet perpendicular-separated
by $\Theta(\delta)$---drives $\dHiso/\dGH$ and $\dBiso/\dGHbij$ to infinity in every $\R^d$,
$d \ge 2$; ambient alignment, Hausdorff or bottleneck, cannot see intrinsic distances. The remaining route to $\dGH$, through $\dGHbij$ itself, fails by the
\emph{multiplicity gap} (\Cref{prop:mult-r1}) already on the line, and the two failures
superimpose on a single \emph{near-collinear} input---one whose spread across a diameter axis is at most a
quarter of its diameter---to defeat their $\min$ (\Cref{prop:barrier}). Together
they rule out aligning the two sets globally, by an isometry or by a bijection, and
\Cref{prop:bij-barrier} below rules out sorting under a single global reflection pattern as well.

The dimension drop is classical: it is the three-point example of \citet[2.12]{ATV2001},
adapted to $\dGH$ by \citet[\S3.1]{Memoli2008}. \Cref{thm:counter-r2} restates it in the normalization used below.

\begin{theorem}[Unbounded ratio in $\R^d$, $d \ge 2$; {\citealp[\S3.1]{Memoli2008}, after \citealp[2.12]{ATV2001}}]
\label{thm:counter-r2}
For every $d \ge 2$ and every $\delta \in (0, \tfrac14)$ there exist finite
$X, Y \subset \R^d$ with $\dHiso(X, Y) = \delta$ and
$\dGH(X, Y) \le 2\delta^2$, hence
\[
\frac{\dHiso(X, Y)}{\dGH(X, Y)} \;\xrightarrow[\delta \to 0]{}\; \infty .
\]
\end{theorem}

\begin{figure}[ht]
\centering
\begin{tikzpicture}[scale=1.0]
  \begin{scope}
    \node[font=\small\itshape, blue!70!black] at (-1.0, 0) {$X_\delta$:};
    \draw[gray!50, very thin] (-0.4, 0) -- (3.4, 0);
    \fill[blue!70!black] (0, 0) circle (2pt) node[above=2pt, font=\small] {$x_1$};
    \fill[blue!70!black] (1.5, 0) circle (2pt) node[above=2pt, font=\small] {$x_2$};
    \fill[blue!70!black] (3, 0) circle (2pt) node[above=2pt, font=\small] {$x_3$};
    \draw[<->, gray!70, thin] (0, -0.7) -- (3, -0.7) node[midway, below, font=\scriptsize] {$D$};
  \end{scope}
  \begin{scope}[xshift=7cm]
    \node[font=\small\itshape, red!70!black] at (-1.0, 0) {$Y_\delta$:};
    \draw[gray!50, very thin] (-0.4, 0) -- (3.4, 0);
    \draw[red!40] (0, 0) -- (1.5, 0.6) -- (3, 0);
    \fill[red!70!black] (0, 0) circle (2pt) node[below left=-1pt, font=\small] {$y_1$};
    \fill[red!70!black] (3, 0) circle (2pt) node[below right=-1pt, font=\small] {$y_3$};
    \fill[red!70!black] (1.5, 0.6) circle (2pt) node[above=2pt, font=\small] {$y_2$};
    \draw[<->, gray!70, thin] (1.65, 0) -- (1.65, 0.6) node[midway, right, font=\scriptsize] {$2\delta$};
    \draw[<->, gray!70, thin] (0, -0.7) -- (3, -0.7) node[midway, below, font=\scriptsize] {$D$};
  \end{scope}
\end{tikzpicture}
\caption{The dimension-drop family of \Cref{thm:counter-r2}: collinear
$X_\delta$ versus $Y_\delta$ with its midpoint lifted by $2\delta$. Here
$\dGH = O(\delta^2/D)$ while $\dHiso = \delta$, so $\dHiso/\dGH =
\Omega(\sqrt{D/\dGH})$.}
\label{fig:counter}
\end{figure}

\begin{proof}
This is the classical construction, scaled to diameter $D$; we reproduce it to fix the exact
constants used below.
Fix $D > 0$ and $\delta \in (0, D/4)$. In $\R^2$, define
\begin{equation}
\label{eq:counter-Xdelta}
X_\delta := \{(0, 0),\ (D/2, 0),\ (D, 0)\},
\end{equation}
\begin{equation}
\label{eq:counter-Ydelta}
Y_\delta := \{(0, 0),\ (D/2, 2\delta),\ (D, 0)\};
\end{equation}
$X_\delta$ is collinear, while $Y_\delta$ is a thin isosceles triangle with
base $D$ and height $2\delta$ (\Cref{fig:counter}). We establish the two bounds
$\dGH(X_\delta, Y_\delta) \le 2\delta^2/D$ and
$\dHiso(X_\delta, Y_\delta) = \delta$, then conclude.

Take the
labeled correspondence $R = \{(x_i, y_i)\}_{i=1,2,3}$ between the points
\eqref{eq:counter-Xdelta}--\eqref{eq:counter-Ydelta}. The pairwise distance
discrepancies are
\[
|d_X(x_1, x_3) - d_Y(y_1, y_3)| = |D - D| = 0,
\qquad
|d_X(x_i, x_2) - d_Y(y_i, y_2)|
= \Bigl|\tfrac{D}{2} - \sqrt{\tfrac{D^2}{4} + 4\delta^2}\Bigr|
\ \ (i = 1, 3),
\]
the last two equal by symmetry. Using $\sqrt{1 + u} - 1 \le u/2$ for
$u \ge 0$ with $u = 16\delta^2/D^2$,
\[
\sqrt{\tfrac{D^2}{4} + 4\delta^2} - \tfrac{D}{2}
= \tfrac{D}{2}\left(\sqrt{1 + \tfrac{16\delta^2}{D^2}} - 1\right)
\le \tfrac{D}{2}\cdot\tfrac{8\delta^2}{D^2}
= \tfrac{4\delta^2}{D},
\]
so $\Dist(R) \le 4\delta^2/D$ and $\dGH(X_\delta, Y_\delta) \le 2\delta^2/D$.

For the $\dHiso$ upper
bound, the translation $z \mapsto z + (0, \delta)$ sends $X_\delta$ to
within distance $\delta$ of $Y_\delta$ in the labeled correspondence, so
$\dHiso(X_\delta, Y_\delta) \le \delta$.
For the matching lower bound, any isometry $T$ maps the collinear $X_\delta$ into a line $L$, so
$\dH(T X_\delta, Y_\delta) \ge \max_y \dist(y, L) \ge \tfrac12\,
\mathrm{width}_{\widehat n}(Y_\delta)$, where $\widehat n \perp L$ and
$\mathrm{width}_{\widehat n}$ is the spread of $Y_\delta$ along $\widehat n$
(the farthest point of $Y_\delta$ from $L$ is at least half that spread). The
minimum over directions is the minimal width of $Y_\delta$, i.e.\ its smallest
altitude; as the base $D$ is the longest side ($\delta < D/4$), that altitude
is the height $2\delta$. Hence $\dHiso(X_\delta, Y_\delta) \ge \tfrac12\cdot
2\delta = \delta$, and both bounds give $\dHiso(X_\delta, Y_\delta) = \delta$.

Both sets have diameter $D$, so
$M := \max(\diam X_\delta, \diam Y_\delta) = D$, and the two bounds give
\[
M\,\dGH(X_\delta, Y_\delta) \le D\cdot\tfrac{2\delta^2}{D} = 2\delta^2
= 2\,\dHiso(X_\delta, Y_\delta)^2,
\quad\text{i.e.}\quad
\dHiso \ge \tfrac{1}{\sqrt 2}\sqrt{M\,\dGH},
\quad \frac{\dHiso}{\dGH} \ge \frac{D}{2\delta}.
\]
Taking $D = 1$ yields the family of the statement, with
$\dHiso/\dGH \ge 1/(2\delta) \to \infty$ as $\delta \to 0$. For $d \ge 3$,
embed $X_\delta, Y_\delta$ in $\R^d$ via
$(a, b) \mapsto (a, b, 0, \ldots, 0)$; $\dGH$ is preserved (it is
intrinsic), and the $\dHiso$ lower bound extends by projection: any isometry
$T$ of $\R^d$ maps the collinear $X_\delta$ into a line $L$, and the
orthogonal projection $P$ onto the plane $\Pi$ containing $Y_\delta$, being
$1$-Lipschitz, maps $L$ to a line or a point of that plane, so
$\dH(TX_\delta, Y_\delta) \ge \max_y \dist(y, L) \ge \max_y \dist(y, P(L))$.
If $P(L)$ is a line, the planar minimal-width argument above bounds this below by $\delta$; if
$L \perp \Pi$, so that $P(L)$ is a single point $p$, then
$\max_y \dist(y, p) \ge \tfrac12\diam Y_\delta = \tfrac12 > \delta$ directly. Either way the bound
is at least $\delta$.
\end{proof}

The same family obstructs the bijective proxy. Its witness correspondence $R$ is itself a
bijection, as is M\'emoli's, so it bounds $\dGHbij$ as well, and the bottleneck route then
overshoots $\dGHbij$ exactly as the Hausdorff route overshoots $\dGH$: a bijection of bottleneck
cost $\varepsilon$ under an isometry $T$ leaves every point of each set within $\varepsilon$ of the
other, so $\dHiso \le \dBiso$.

\begin{corollary}[Unbounded ratio for the bottleneck proxy; implicit in {\citealp[\S3.1]{Memoli2008}}]
\label{cor:counter-bij}
The family of \Cref{thm:counter-r2} has $|X| = |Y|$ and satisfies
$\dGHbij(X, Y) \le 2\delta^2$ and $\dBiso(X, Y) \ge \dHiso(X, Y) = \delta$, hence
\[
\frac{\dBiso(X, Y)}{\dGHbij(X, Y)} \;\xrightarrow[\delta \to 0]{}\; \infty
\qquad\text{in every } \R^d,\ d \ge 2 .
\]
\end{corollary}

The natural remedy for the dimension drop---matching points by a
\emph{bijection}---fails for an unrelated reason. By \eqref{eq:dGHbij},
$\tfrac12\Dist(\sigma) \ge \dGHbij$ for every bijection $\sigma$, and $\dGHbij \ge \dGH$ since
bijections are correspondences; so a bijection-valued estimator, reporting $\tfrac12\Dist(\sigma)$,
never falls below $\dGHbij$. But $\dGHbij$ exceeds $\dGH$ by an unbounded factor \emph{already on the
line}, so no such estimator---sorted-order matching along an axis, or any other---can approximate
$\dGH$.

\begin{proposition}[Multiplicity gap on $\R^1$]
\label{prop:mult-r1}
For every $\varepsilon \in (0, \tfrac14)$ there exist four-point sets
$X, Y \subset \R$ with $\dGH(X, Y) \le \varepsilon$ and
$\dGHbij(X, Y) \ge \tfrac12(1 - 3\varepsilon)$, hence
\[
\frac{\dGHbij(X, Y)}{\dGH(X, Y)} \;\xrightarrow[\varepsilon \to 0]{}\; \infty .
\]
\end{proposition}

\begin{proof}
Label $X = \{x_1, x_2, x_3, x_4\}$ at
positions $0, \varepsilon, 1, 1+\varepsilon$ and
$Y = \{y_1, y_2, y_3, y_4\}$ at $0, \varepsilon, 2\varepsilon, 1$. The
correspondence
\[
R = \{(x_1, y_1),\ (x_2, y_2),\ (x_2, y_3),\ (x_3, y_4),\ (x_4, y_4)\}
\]
is surjective onto both sides. Its distortion is the maximum, over the
$\binom{5}{2}$ pairs of elements of $R$, of $|d_X - d_Y|$. Every $X$-coordinate
difference and every $Y$-coordinate difference lies in $\{0, \varepsilon, 2\varepsilon,
1-2\varepsilon, 1-\varepsilon, 1, 1+\varepsilon\}$, and the discrepancy is maximized at the
pair $(x_2, y_3),\,(x_4, y_4)$, where
$|d_X(x_2, x_4) - d_Y(y_3, y_4)| = |1 - (1 - 2\varepsilon)| = 2\varepsilon$; every other pair
differs by at most this (e.g.\ $|d_X(x_3,x_4)-d_Y(y_4,y_4)| = \varepsilon$ and
$|d_X(x_2,x_2)-d_Y(y_2,y_3)| = \varepsilon$). Hence $\Dist(R) \le
2\varepsilon$ and $\dGH \le \varepsilon$.

Of the four points of $Y$, three lie
in $[0, 2\varepsilon]$ (namely $0, \varepsilon, 2\varepsilon$) and only
one (namely $1$) lies outside it. So under any bijection $\sigma$
exactly one point of $X$ maps to $y = 1$ and the other three map into
$[0, 2\varepsilon]$. That one point belongs to one of the two tight
pairs $\{x_1, x_2\} = \{0, \varepsilon\}$ or $\{x_3, x_4\} = \{1, 1 +
\varepsilon\}$ (each at distance $\varepsilon$), and its partner then
maps into $[0, 2\varepsilon]$; that pair is therefore split, its two
images lying at distance $\ge 1 - 2\varepsilon$. Hence
$\Dist(\sigma) \ge (1 - 2\varepsilon) - \varepsilon = 1 - 3\varepsilon$,
and $\dGHbij \ge \tfrac12(1 - 3\varepsilon)$.
\end{proof}

The mechanism is \emph{multiplicity mismatch}: $X$ carries two tight pairs while $Y$ carries a tight
triple and a singleton, so a correspondence absorbs the size difference by matching many-to-one
whereas a bijection must split a tight pair across the unit gap. This multiplicity gap and the
dimension drop act on disjoint structure---cluster sizes and perpendicular spread,
respectively---so they superimpose on a single near-collinear input.

\begin{proposition}[Combined barrier]
\label{prop:barrier}
For every $\delta \in (0, \tfrac18]$ there is a near-collinear pair
$X, Y \subset \R^2$ with $\dHiso(X, Y) \ge \delta$ and
$\dGH(X, Y) \le 2\delta^2$, while $\dGHbij(X, Y) \ge \tfrac14 - \tfrac32\delta^2$. Hence
\[
\frac{\min\bigl(\dHiso(X, Y),\ \dGHbij(X, Y)\bigr)}{\dGH(X, Y)}
\;\xrightarrow[\delta \to 0]{}\; \infty .
\]
\end{proposition}

\begin{proof}
Take the collinear set with its near-collinear
partner
\begin{align*}
X &= \{(0,0),\ (\delta^2, 0),\ (\tfrac12, 0),\ (1, 0),\ (1+\delta^2, 0)\},\\
Y &= \{(0,0),\ (\delta^2, 0),\ (2\delta^2, 0),\ (\tfrac12, 2\delta),\ (1, 0)\}.
\end{align*}
First $\dGH \le 2\delta^2$, via the \emph{monotone} correspondence: match each $X$-cluster to the
$Y$-cluster near $0$ (resp.\ $1$) preserving left-to-right order---so the extra $Y$-point $(2\delta^2,0)$ pairs with
$(\delta^2,0)$, never $(0,0)$ with $(\delta^2,0)$---and the midpoint $(\tfrac12, 0)$ to the apex
$(\tfrac12, 2\delta)$. A matched pair avoiding the apex has each point within $2\delta^2$ of its partner, so
its two distances differ by $\le 4\delta^2$. A pair through the apex and a left-cluster point
$x' \leftrightarrow (a,0)$ has, by monotonicity, abscissa $\pi(x') \le a$ (write $\pi$ for the first
coordinate): the $X$-distance
$\tfrac12 - \pi(x')$ exceeds the longitudinal part $\tfrac12 - a$ of the $Y$-distance by at most
$a - \pi(x') \le 2\delta^2$, while the $Y$-distance exceeds it by exactly
\[
\sqrt{(\tfrac12 - a)^2 + 4\delta^2} - (\tfrac12 - a)
= \frac{4\delta^2}{\sqrt{(\tfrac12-a)^2+4\delta^2} + (\tfrac12 - a)} \;\le\; 4\delta^2,
\]
the denominator being $\ge 1$ (equivalent to $4\delta^2 \ge 2a$, and $a \le 2\delta^2$); so the
pair's two distances differ by $\le 4\delta^2$, and symmetrically for a right-cluster point
($\pi(x') \ge a$, the roles of the two bounds exchanged, the denominator there $\ge 1$ outright).
Hence $\dGH \le 2\delta^2$.

Next $\dHiso \ge \delta$: any isometry maps the collinear $X$ to a line, and $Y$ has minimal width
$2\delta$ (the four $x$-axis points against the apex at height $2\delta$), so the minimal-width
argument of \Cref{thm:counter-r2} gives $\dHiso \ge \delta$. The exact value and any upper bound are
immaterial; the lower bound $\delta$ is all we use.

Finally $\dGHbij \ge \tfrac14 - \tfrac32\delta^2$, by a multiplicity pigeonhole as in
\Cref{prop:mult-r1}. The three $Y$-points $(0,0), (\delta^2,0), (2\delta^2,0)$
are pairwise within $2\delta^2$, whereas no three points of $X$ are mutually close:
each maximal cluster of $X$ has at most two points (the pairs near $0$ and near $1$,
with $\tfrac12$ isolated), so any three points of $X$ contain a pair at distance
$\ge \tfrac12 - \delta^2$. Under any bijection $\sigma$ the three preimages of
$\{(0,0),(\delta^2,0),(2\delta^2,0)\}$ are three $X$-points, hence contain such a
far pair, whose images lie within $2\delta^2$; so
$\Dist(\sigma) \ge (\tfrac12 - \delta^2) - 2\delta^2 = \tfrac12 - 3\delta^2$
and $\dGHbij \ge \tfrac14 - \tfrac32\delta^2$.

Both branches are thus at least $\delta$ (using $\tfrac14 - \tfrac32\delta^2 \ge \delta$ for
$\delta \le \tfrac18$) while $\dGH \le 2\delta^2$, which is the displayed divergence.
\end{proof}

Isometries and bijections fail. A third global device, one reflection pattern for the whole
input, fails as well. The \emph{global sorted estimator} $\widehat\rho_{\mathrm{sort}}$ sorts $X$ and
$Y$ each along a direction (ties broken by the perpendicular coordinate), matches by rank, and keeps
the best over all directions. It returns a bijection's half-distortion, hence an upper bound on
$\dGHbij$. But a single global sort cannot reflect distinct sub-clusters independently, so a
reflected sub-cluster makes it overshoot $\dGHbij$ by $\Omega(\sqrt{\Delta/\dGHbij})$.

\begin{proposition}[Reflection barrier for the bijective distance]
\label{prop:bij-barrier}
For every $\Delta > 0$ and $\eta \in (0, \Delta/100]$ there is a near-collinear pair
$X, Y \subset \R^2$ with $|X| = |Y| = 6$, $\diam(X) \in [\Delta, 2\Delta]$, and
$\dGHbij(X, Y) \le \eta$, on which
\[
\widehat\rho_{\mathrm{sort}}(X, Y) \;\ge\; \tfrac16\sqrt{\Delta\,\eta}
\;\ge\; \tfrac16\sqrt{\Delta\,\dGHbij(X,Y)}
\;=\; \Omega\left(\sqrt{\tfrac{\Delta}{\dGHbij(X,Y)}}\right)\cdot \dGHbij(X,Y).
\]
The same gap holds for every estimator minimizing over a subset of these sorts, such as a
common axis for both sets.
\end{proposition}

\begin{proof}
Set $h = \tfrac13\sqrt{\Delta\eta} \le \Delta/30$ and place two vertical three-point clusters at
longitudinal positions $0$ and $\Delta$ (\Cref{fig:bij-barrier}): $X = \{(0,p),(\Delta,p) : p \in \{0,h,3h\}\}$,
with $Y$ the same but the $\Delta$-cluster's perpendiculars negated. The diameter axis runs within
$O(h/\Delta)$ of the cluster-separation direction, and the perpendicular spread is less than $6h \le \Delta/5$
(two $3h$-tall clusters straddling that axis), so both are near-collinear. The identity bijection is
intra-cluster isometric (negation preserves $\{h,2h,3h\}$), and its largest cross-distortion is at
$(0,3h),(\Delta,3h)$---matched to $(0,3h),(\Delta,-3h)$---where the distance grows from $\Delta$ to
$\sqrt{\Delta^2+36h^2}$, by at most $36h^2/2\Delta = 18h^2/\Delta$; so $\dGHbij \le 9h^2/\Delta = \eta$.

Conversely, $\widehat\rho_{\mathrm{sort}}$ tries \emph{every} sort direction, and each fails. Take any
direction and its sorted bijection $\sigma$; we may assume $\Dist(\sigma) < \Delta/2$ (else
$\Dist(\sigma) \ge \Delta/2 > h$ already), so $\sigma$ is block-preserving (two $X$-points at distance
$\le 3h$ cannot split across $Y$-clusters at distance $\ge \Delta - 6h$). Inside a vertical cluster
the rank order is by perpendicular coordinate up to one sign---that of the direction's perpendicular
component; for the one direction with no perpendicular component the tie-break of the definition
supplies the same common sign---so the direction gives \emph{both} parallel $X$-clusters, and \emph{both} $Y$-clusters, a
\emph{common} ordering; no single global sort can reflect one cluster relative to a parallel one. The
perpendicular-gap sequences are then $(h, 2h)$ for both $X$-clusters and $Y$'s $0$-cluster but $(2h, h)$
for $Y$'s reflected $\Delta$-cluster, so for either sign and either block assignment exactly one matched
cluster-pair is mismatched, forcing a gap-$h$ pair onto a gap-$2h$ pair and $\Dist(\sigma) \ge h$. As
this holds for every direction, $\widehat\rho_{\mathrm{sort}} \ge h/2 = \tfrac16\sqrt{\Delta\eta} =
\Omega(\sqrt{\Delta/\dGHbij}) \cdot \dGHbij$. A variant minimizing over a subset of these sorts
returns the half-distortion of one of them, every one of which fails at $h$, so it overshoots by
at least the same factor. (The restriction to \emph{these} sorts is load-bearing: the common
perpendicular tie-break is what forces one ordering on parallel clusters; a sort free to break
ties per cluster could reflect the $\Delta$-cluster alone and evade the bound.)
\end{proof}

\begin{figure}[ht]
\centering
\begin{tikzpicture}[x=1cm,y=1cm,
  Xpt/.style={blue!70!black, fill=blue!70!black, circle, inner sep=1.3pt},
  Ypt/.style={red!70!black, fill=red!70!black, circle, inner sep=1.3pt},
  ax/.style={-{Latex[length=1.6mm]}, gray!60, semithick},
  gp/.style={{Latex[length=1.1mm]}-{Latex[length=1.1mm]}, gray!80}]

\begin{scope}
  \node[font=\itshape, blue!70!black] at (-1.05,0.6) {$X$};
  \draw[ax] (-0.6,0) -- (3.3,0) node[below=1pt, font=\scriptsize, black] {$\pi$};
  \draw[ax] (-0.6,0) -- (-0.6,1.45) node[above=-1pt, font=\scriptsize, black] {$\pi^\perp$};
  \foreach \x in {0,2.4}{
    \draw[gray!30, dotted] (\x,0)--(\x,1.3);
    \foreach \y in {0,0.4,1.2} \node[Xpt] at (\x,\y){};
  }
  \node[font=\scriptsize, below=3pt] at (0,0) {$0$};
  \node[font=\scriptsize, below=3pt] at (2.4,0) {$\Delta$};
  \draw[gp] (2.63,0.03)--(2.63,0.37) node[midway, right=1pt, font=\scriptsize] {$h$};
  \draw[gp] (2.63,0.43)--(2.63,1.17) node[midway, right=1pt, font=\scriptsize] {$2h$};
\end{scope}

\begin{scope}[xshift=5cm]
  \node[font=\itshape, red!70!black] at (-1.05,0.6) {$Y$};
  \draw[ax] (-0.6,0) -- (3.3,0) node[below=1pt, font=\scriptsize, black] {$\pi$};
  \draw[gray!30, dotted] (0,0)--(0,1.3);
  \foreach \y in {0,0.4,1.2} \node[Ypt] at (0,\y){};
  \draw[gray!30, dotted] (2.4,0)--(2.4,-1.3);
  \foreach \y in {0,-0.4,-1.2} \node[Ypt] at (2.4,\y){};
  \node[font=\scriptsize, below=3pt] at (0,0) {$0$};
  \node[font=\scriptsize, above=3pt] at (2.4,0) {$\Delta$};
  \draw[gp] (2.63,-0.03)--(2.63,-0.37) node[midway, right=1pt, font=\scriptsize] {$h$};
  \draw[gp] (2.63,-0.43)--(2.63,-1.17) node[midway, right=1pt, font=\scriptsize] {$2h$};
\end{scope}
\end{tikzpicture}
\caption{The reflection barrier (\Cref{prop:bij-barrier}): both clusters of
$X$ carry the perpendicular gap pattern $(h, 2h)$, while $Y$ reflects its $\Delta$-cluster to $(2h, h)$.
No global sort, in any direction and with any reflection, reproduces both, so
$\widehat\rho_{\mathrm{sort}} \ge \tfrac16\sqrt{\Delta\eta}$ though $\dGHbij \le \eta$.}
\label{fig:bij-barrier}
\end{figure}

\begin{remark}[Signs must be local]
\label{rem:per-node-signs}
One reflection sign cannot serve both blocks below. Place two vertical
blocks of pattern $(0, h, 3h)$ at longitudinal positions $0$ and $L$, joined by a dense run at height
$\tfrac32 h$ with longitudinal gaps $2h$, and let $Y$ reflect only the block at $0$ across $\tfrac32 h$.
Block--run distances are preserved, since the run lies on the mirror axis, and only block--block
pairs distort: $\dGH \le \dGHbij = O(h^2/L)$, realized by arranging the two blocks
\emph{oppositely}. An estimator that orders both blocks by height under a single sign leaves one
block's bottom-to-middle distance at $h$ where $Y$ has $2h$. Its distortion is at least $h$, an
overshoot of $\Omega(L/h)$. The algorithm of
\Cref{sec:main} therefore carries one sign per window.
\end{remark}

\section{Frames and windows}
\label{sec:frame}

\Cref{sec:dimension-drop} leaves the following situation. A near-isometry of a planar set need not be
near an isometry: as the dimension drop shows, its perpendicular coordinates can drift by
$\Theta(\sqrt{\mu D})$, where $\mu$ is the distortion and $D$ the diameter, while every pairwise
distance changes by at most $\mu$. What such a map cannot do is move points \emph{along} the
diameter. This section makes that precise and builds the algorithm's local description on it.
Reading each set in the frame spanned by a diameter pair and its images, the longitudinal coordinate
of every point is pinned to within $O(\mu)$ (\Cref{lem:frame}), so the height is the only unknown.
The height is then recovered locally: in a window around each point whose height spread is at most
a fixed multiple of its length, a three-point argument (\Cref{lem:three}) shows that an optimal
correspondence acts on heights as a translation, up to one reflection (\Cref{lem:window}).
\Cref{sec:locality} turns these window-local descriptions into one correspondence.

Throughout, $R^*$ is an optimal correspondence between finite $X, Y \subset \R^2$ and
$\mu := \Dist(R^*) = 2\,\dGH(X,Y)$. A \emph{partner} of $z$ is any $z'$ with $(z, z') \in R^*$,
and symmetrically for points of $Y$. Let $D := \diam X > 0$.

\paragraph{Parameters.}
The analysis has three parameters: a \emph{threshold} $\nu_0 \in (0, \tfrac12]$, a \emph{window
scale} $a > 0$ and a \emph{spread cap} $\beta > 0$; the algorithm adds a \emph{ladder ratio}
$\varepsilon > 0$. Every other constant is a function of these, defined where it first arises and
collected in \Cref{tab:constants}, which also states the three requirements the parameters must meet
and the values used. We assume $\mu \le \nu_0 D$ throughout; otherwise the correspondence
$X \times Y$ already has $\Dist \le D + \mu \le (1 + 1/\nu_0)\,\mu$.

\paragraph{Frames.}
Fix a diameter pair $(x_\ell, x_r)$ of $X$ and read $X$ in its frame: origin $x_\ell$, first axis
along $x_\ell x_r$, coordinates $(\pi, h)$---the \emph{longitudinal coordinate} and the
\emph{height}. Read $Y$ in the frame of an ordered pair $(y, y')$ of distinct points of $Y$: origin
$y$, first axis along $y y'$. The algorithm tries every such pair; the analysis uses the one
formed by partners of $x_\ell$ and $x_r$, which are distinct since
$|y - y'| \ge D - \mu > 0$.

\begin{lemma}[Frame]
\label{lem:frame}
Let $(x_\ell, y), (x_r, y') \in R^*$, put $\nu := \mu/D \le \nu_0$, and
\[
c_0 := \frac{6 + \nu_0}{2(1 - \nu_0)}, \qquad c_r := \frac{4 + 2\nu_0}{1 - \nu_0}, \qquad
\nu_1 := \frac{\nu}{1 - \nu_0} .
\]
\begin{enumerate}
\item[(i)] $|\pi(z') - \pi(z)| \le c_0\,\mu$ for every $(z, z') \in R^*$.
\item[(ii)] For $(p, p'), (q, q') \in R^*$, the longitudinal differences $\Delta\pi := \pi(q) - \pi(p)$
  and $\Delta\pi' := \pi(q') - \pi(p')$ satisfy
  $|\Delta\pi' - \Delta\pi| \le c_r\,\mu + \nu_1 \min(|\Delta\pi|, |\Delta\pi'|)$.
\end{enumerate}
\end{lemma}

\begin{proof}
With $p := |z - x_\ell|$, $q := |z - x_r|$ and $A := p^2 - q^2$, the projection is
$\pi(z) = A/(2D) + D/2$. Likewise $\pi(z') = A'/(2D') + D'/2$, with $p', q'$ and
$D' := |y - y'|$. So
\[
\pi(z') - \pi(z) \;=\; \frac{A' - A}{2D'} + A\Bigl(\frac{1}{2D'} - \frac{1}{2D}\Bigr) + \frac{D' - D}{2}.
\]
Each primed quantity is within $\mu$ of its unprimed one, $p, q \in [0, D]$, and
$D' \ge (1 - \nu)D$. So the three terms are at most
$\frac{2 + \nu}{1 - \nu}\mu$ (since $|A' - A| \le 2\mu(p + q) + 2\mu^2$), $\frac{\mu}{2(1 - \nu)}$ (since
$|A| \le D^2$), and $\frac{\mu}{2}$. Their sum is $\frac{6 + \nu}{2(1 - \nu)}\mu \le c_0\mu$, as
$\nu \le \nu_0$. This is (i).

For (ii), the term $D'/2$ cancels in a difference. Write $A_p, A_q$ and $A'_p, A'_q$ for the values
of $A$ and $A'$ at the two pairs. Then
\[
\Delta\pi' - \Delta\pi \;=\; \frac{(A'_q - A_q) - (A'_p - A_p)}{2D'}
+ \frac{A_q - A_p}{2D}\cdot\frac{D - D'}{D'} .
\]
The first term is at most $2 \cdot \frac{2 + \nu}{1 - \nu}\mu \le c_r\mu$. The second is
$\Delta\pi\,(D - D')/D'$, at most $|\Delta\pi|\,\mu/D' \le \frac{\nu}{1 - \nu}|\Delta\pi| \le \nu_1|\Delta\pi|$.
Splitting instead as $\frac{(A'_q - A'_p) - (A_q - A_p)}{2D} + \Delta\pi'\,\frac{D - D'}{D}$ bounds
the two terms by $(4 + 2\nu)\mu \le c_r\mu$ and $\nu|\Delta\pi'| \le \nu_1|\Delta\pi'|$.
\end{proof}

The lemma is symmetric in $X$ and $Y$: it bounds the change of longitudinal coordinate across every
pair of $R^*$. Only the height is left unpinned, and \Cref{thm:counter-r2} shows it can move by
$\Theta(\sqrt{\mu D})$.

\paragraph{Windows.}
The algorithm guesses $\mu$. Fix a \emph{guess} $\widehat\mu \ge \mu$, and put $g := c_0$.
For $x \in X$ and $\lambda > 0$ write $W_\lambda(x) := \{z \in X : |\pi(z) - \pi(x)| \le \lambda\}$,
and $s(Z) := \max_{Z} h - \min_{Z} h$ for the height spread of a set $Z$.

\begin{definition}[Windows and gates]
\label{def:window}
The \emph{window} of $x$ is $W(x) := W_{\lambda(x)}(x)$, where $\lambda(x)$ is the least
$\lambda \ge a\widehat\mu$ with $s(W_\lambda(x)) \le \beta\lambda$. The window is
\emph{short} if $\lambda(x) = a\widehat\mu$ and \emph{grown} otherwise. The \emph{gate} of $z \in X$ is
$G(z) := \{y \in Y : |\pi(y) - \pi(z)| \le g\widehat\mu\}$. Windows of points of $Y$, and gates in $X$,
are defined symmetrically.
\end{definition}

The least $\lambda$ exists: $\lambda \mapsto s(W_\lambda(x))$ is a nondecreasing step function,
continuous from the right because $W_\lambda(x)$ is defined by a non-strict inequality, and the
condition holds once $\beta\lambda \ge s(X)$. If the window is grown, the condition fails for every
$\lambda \in [a\widehat\mu, \lambda(x))$, so $s(W(x)) \ge s(W_\lambda(x)) > \beta\lambda$ for all such
$\lambda$. Hence $s(W(x)) \ge \beta\lambda(x)$, and so $s(W(x)) = \beta\lambda(x) > \beta a\widehat\mu$:
every window is either short, or has height spread $\beta$ times its half-length. Windows are thus
allowed to be taller than they are long, which keeps the longitudinal part of~\eqref{eq:height}
small.

The one estimate behind everything below concerns two pairs $(p, p'), (q, q') \in R^*$. Write
$d, d'$ for the distances $|p - q|, |p' - q'|$, and $\Delta\pi, \Delta h$ and
$\Delta\pi', \Delta h'$ for the coordinate differences in the two frames. Since
$d^2 = \Delta\pi^2 + \Delta h^2$ in both frames, \Cref{lem:frame}(ii) gives
\begin{equation}
\label{eq:height}
\bigl|\Delta h'^2 - \Delta h^2\bigr| \;\le\; \bigl|d'^2 - d^2\bigr| + \bigl|\Delta\pi'^2 - \Delta\pi^2\bigr|
\;\le\; \mu(2d + \mu) + e_\pi\bigl(2|\Delta\pi| + e_\pi\bigr),
\qquad e_\pi := c_r\mu + \nu_1|\Delta\pi| .
\end{equation}
The same holds with $d'$ and $\Delta\pi'$ on the right side, which is the form used for windows of $Y$.

All rigidity comes from three points: two anchors and a third point, read through
\eqref{eq:height}.

\begin{lemma}[Three points]
\label{lem:three}
Let $s \ge 0$, $u \in [0, s]$, and let $s' \ge 0$ and $u'$ be reals with
$|s'^2 - s^2|, |u'^2 - u^2|, |(s' - u')^2 - (s - u)^2| \le E$. Then $|u' - u| \le 2\sqrt E$, and
$|u' - u| \le 3E/s$ if $s > 2\sqrt E$.
\end{lemma}

\begin{proof}
Put $v := s - u$ and $v' := s' - u'$. Note $|s' - s| \le \sqrt E$, and $|s' - s| \le E/s$ if $s > 0$.

\emph{Case $u', v' \ge 0$.} Then $|u' - u| \le \sqrt E$, since $|u' - u|^2 \le |u' - u|(u' + u) \le E$.
Also $|u' - u| \le E/u$ and $|v' - v| \le E/v$ when $u, v > 0$. One of $u, v$ is at least $s/2$; if it is $v$, then
$|u' - u| \le |s' - s| + |v' - v|$. Either way $|u' - u| \le 3E/s$.

\emph{Case $u' < 0$.} Then $v' = s' + |u'| \le \sqrt{v^2 + E}$ and
$s' \ge \sqrt{\max(s^2 - E, 0)}$, so
\[
|u' - u| = |u'| + u \;\le\; u + \sqrt{(s - u)^2 + E} - \sqrt{\max(s^2 - E, 0)}
\;\le\; s + \sqrt E - \sqrt{\max(s^2 - E, 0)} \;\le\; 2\sqrt E,
\]
the middle step because the expression is nondecreasing in $u$. Now let $s > 2\sqrt E$, so
$s' \ge s - E/s > 3s/4$. If $u > s/2$, then $v'^2 \ge s'^2 > 9s^2/16 > s^2/4 + E > v^2 + E$, which is
impossible. So $v \ge s/2$, and using $\sqrt{v^2 + E} \le v + E/(2v)$,
$|u'| + u \le E/(2v) + E/s \le 2E/s$.

\emph{Case $v' < 0$.} Then $u' > s' \ge 0$, so $|u' - u| \le \sqrt E$ as in the first case. If
$s > 2\sqrt E$, the previous argument with $u$ and $v$ exchanged gives $u \ge s/2$, so
$|u' - u| \le E/u \le 2E/s$.
\end{proof}

The constants of the next lemma are
\begin{equation}
\label{eq:window-consts}
\begin{aligned}
e_s &:= c_r + \frac{2a\nu_0}{1 - \nu_0}, &\qquad E_s &:= 2a\sqrt{4 + \beta^2} + 1 + e_s(4a + e_s), \\
e_g &:= c_r + \frac{2(1 + \nu_0)}{(1 - \nu_0)\beta}, &\qquad
C_1 &:= 2\sqrt{1 + 4/\beta^2} + \frac{4c_r}{\beta} + \frac{8(1 + \nu_0)}{(1 - \nu_0)\beta^2}
+ \frac{1 + e_g^2}{\beta a}, \\
\rho &:= \max\bigl(2\sqrt{E_s},\, 3C_1\bigr).
\end{aligned}
\end{equation}

\begin{lemma}[Window rigidity]
\label{lem:window}
Assume $\beta a \ge 4C_1$. For every $x \in X$ there are $\varepsilon_x \in \{\pm 1\}$ and $b_x \in \R$ with
\[
\bigl|h(z') - (\varepsilon_x h(z) + b_x)\bigr| \;\le\; \rho\,\widehat\mu
\qquad \text{for every } (z, z') \in R^* \text{ with } z \in W(x).
\]
\end{lemma}

\begin{proof}
Let $s := s(W(x))$, and take $a_0, b_0 \in W(x)$ of least and greatest height, with partners
$a_0', b_0'$. Put $\varepsilon_x := \operatorname{sign}\bigl(h(b_0') - h(a_0')\bigr)$ (either sign if
this is $0$), $b_x := h(a_0') - \varepsilon_x h(a_0)$ and $s' := \varepsilon_x\bigl(h(b_0') - h(a_0')\bigr) \ge 0$.
For $(z, z') \in R^*$ with $z \in W(x)$ put $u := h(z) - h(a_0) \in [0, s]$ and
$u' := \varepsilon_x\bigl(h(z') - h(a_0')\bigr)$, so that
$|h(z') - (\varepsilon_x h(z) + b_x)| = |u' - u|$. The pairs $(a_0, b_0)$, $(a_0, z)$ and $(z, b_0)$ lie
in the window, so \eqref{eq:height} gives the hypotheses of \Cref{lem:three} with $E$ any bound
on the right side of \eqref{eq:height} over pairs of $W(x)$.

\emph{Short window.} Pairs have $|\Delta\pi| \le 2a\widehat\mu$, $|\Delta h| \le \beta a\widehat\mu$ and
$d \le \sqrt{4 + \beta^2}\,a\widehat\mu$. With $\mu \le \widehat\mu$ and $\nu_1 \le \nu_0/(1 - \nu_0)$
this gives $e_\pi \le e_s\widehat\mu$ and $E \le E_s\widehat\mu^2$, and \Cref{lem:three} gives
$|u' - u| \le 2\sqrt{E_s}\,\widehat\mu \le \rho\widehat\mu$.

\emph{Grown window.} Now $s = \beta\lambda(x) > \beta a\widehat\mu \ge \beta a\mu$. Pairs have
$|\Delta\pi| \le 2s/\beta$, $|\Delta h| \le s$ and $d \le s\sqrt{1 + 4/\beta^2}$. A height spread is
at most a diameter, so $s \le D + \mu \le (1 + \nu_0)D$, and hence
$e_\pi \le e_g\,\mu$ and $\nu_1|\Delta\pi|^2 \le 4(1 + \nu_0)\,s\mu/((1 - \nu_0)\beta^2)$. Expanding
the right side of \eqref{eq:height} as $2d\mu + 2e_\pi|\Delta\pi| + \mu^2 + e_\pi^2$ and using
$\mu^2 \le s\mu/(\beta a)$ for the last two terms gives $E \le C_1 s\mu$.
Since $s > \beta a\mu \ge 4C_1\mu$, we have $s > 2\sqrt{C_1 s\mu} \ge 2\sqrt E$, and \Cref{lem:three} gives
$|u' - u| \le 3E/s \le 3C_1\mu \le \rho\widehat\mu$.
\end{proof}

The same statement holds for windows of $Y$, with the roles of $X$ and $Y$ exchanged, since
\eqref{eq:height} is symmetric.

\section{Locality and signs}
\label{sec:locality}

\begin{definition}[Admissibility]
\label{def:admissible}
Let $\tau := 2\rho$. A point $y \in G(x)$ is \emph{$\varepsilon$-admissible} for $x$ if, with
$b := h(y) - \varepsilon h(x)$, every $z \in W(x)$ has some $y' \in G(z)$ with
$|h(y') - (\varepsilon h(z) + b)| \le \tau\widehat\mu$. Write $A_\varepsilon(x)$ for the set of such $y$.
The sets $B_\varepsilon(y) \subseteq X$ for $y \in Y$ are defined symmetrically.
\end{definition}

By \Cref{lem:frame,lem:window}, every partner $x'$ of $x$ lies in $A_{\varepsilon_x}(x)$: for $y'$
take any partner $z'$ of $z$, which lies in $G(z)$ because $c_0\mu \le g\widehat\mu$. Indeed, with
$b := h(x') - \varepsilon_x h(x)$, \Cref{lem:window} applied to $(x, x')$ gives
$|b - b_x| \le \rho\widehat\mu$, so $|h(z') - (\varepsilon_x h(z) + b)| \le 2\rho\widehat\mu = \tau\widehat\mu$.
The converse is the heart of the argument: every admissible target is, up to $O(\widehat\mu)$, a
partner. Its constants are
\begin{equation}
\label{eq:locality-consts}
\gamma := g + c_0 = 2c_0, \qquad H := 4\rho + \frac{\beta a\gamma}{a - \gamma}, \qquad
\Lambda := \sqrt{\gamma^2 + H^2}.
\end{equation}

\begin{lemma}[Locality]
\label{lem:locality}
Assume $a > \gamma$. If $y \in A_{\varepsilon_x}(x)$ and $j$ is a
partner of $y$, then $|\pi(j) - \pi(x)| \le \gamma\widehat\mu$ and
$|h(j) - h(x)| \le H\widehat\mu$. In particular $\|j - x\| \le \Lambda\widehat\mu$.
\end{lemma}

\begin{proof}
Write $f(z) := \varepsilon_x h(z)$ and $\Delta := b - b_x$, with $b$ as in
\Cref{def:admissible}. Suppose $\Delta > (\tau + \rho)\widehat\mu$. The case
$\Delta < -(\tau + \rho)\widehat\mu$ is symmetric, with the chain below running downward.

\emph{Step 1: every window point has a higher partner.} For $z \in W(x)$, take the $y'$ of the
definition and a partner $z'$ of $y'$. By \Cref{lem:frame},
$|\pi(z') - \pi(z)| \le g\widehat\mu + c_0\mu \le \gamma\widehat\mu$. If $z' \in W(x)$, \Cref{lem:window}
gives $f(z') = f(z) + \Delta \pm (\tau + \rho)\widehat\mu$, so $f(z') > f(z)$.

\emph{Step 2: the chain must leave the window.} Start at $z_0 := x$ and follow such partners while
they stay in $W(x)$. Since $f$ strictly increases, no point repeats, so some $z_k$ has its partner
outside $W(x)$. That partner is more than $\lambda := \lambda(x)$ from $\pi(x)$ and within
$\gamma\widehat\mu$ of $\pi(z_k)$, and each step moves at most $\gamma\widehat\mu$. Hence
$k \ge (\lambda - \gamma\widehat\mu)/(\gamma\widehat\mu)$.

\emph{Step 3: the climb exceeds the spread.} All of $z_0, \dots, z_k$ lie in $W(x)$, so
\[
\beta\lambda \;\ge\; s(W(x)) \;\ge\; f(z_k) - f(z_0) \;\ge\; \frac{\lambda - \gamma\widehat\mu}{\gamma\widehat\mu}\,
\bigl(\Delta - (\tau + \rho)\widehat\mu\bigr).
\]
Since $\lambda \ge a\widehat\mu > \gamma\widehat\mu$ and $\lambda\gamma/(\lambda - \gamma\widehat\mu)$ decreases in $\lambda$, this
gives $\Delta \le \bigl(\tau + \rho + \beta a\gamma/(a - \gamma)\bigr)\,\widehat\mu$.
Hence $|\Delta| \le \bigl(\tau + \rho + \beta a\gamma/(a - \gamma)\bigr)\widehat\mu$ in all cases.

\emph{Step 4: read off $j$.} By \Cref{lem:frame} and $y \in G(x)$,
$|\pi(j) - \pi(x)| \le \gamma\widehat\mu \le \lambda$, so $j \in W(x)$. \Cref{lem:window} gives
$f(j) + b_x = h(y) \pm \rho\widehat\mu = f(x) + b \pm \rho\widehat\mu$, so
$|h(j) - h(x)| \le |\Delta| + \rho\widehat\mu \le (\tau + 2\rho + \beta a\gamma/(a - \gamma))\widehat\mu = H\widehat\mu$,
as $\tau = 2\rho$. Finally $\sqrt{\gamma^2 + H^2} = \Lambda$.
\end{proof}

\begin{figure}[ht]
\centering
\begin{tikzpicture}[x=0.9cm,y=0.55cm]
  \draw[gray!60] (-4.4,-0.4) rectangle (4.4,4.6);
  \node[gray, font=\scriptsize, anchor=south west] at (-4.4,4.6)
    {$W(x)$: length $2\lambda$, height spread $\le \beta\lambda$};
  \foreach \px/\py in {0/0, 0.7/1.1, 1.5/2.1, 2.2/3.0, 3.1/3.9} { \fill (\px,\py) circle (2pt); }
  \node[left, font=\scriptsize] at (0,0.1) {$x = z_0$};
  \foreach \a/\b/\c/\d in {0/0/0.7/1.1, 0.7/1.1/1.5/2.1, 1.5/2.1/2.2/3.0, 2.2/3.0/3.1/3.9} {
    \draw[-{Latex[length=1.5mm]}, blue!60] (\a,\b) -- (\c,\d); }
  \draw[-{Latex[length=1.5mm]}, red!70!black, dashed] (3.1,3.9) -- (4.9,4.9);
  \fill[red!70!black] (4.9,4.9) circle (2pt);
  \node[right, font=\scriptsize, red!70!black] at (4.95,4.9) {partner outside};
  \node[font=\scriptsize, blue!60!black, align=left] at (-2.2,2.6)
    {each step: $\le \gamma\widehat\mu$ along,\\ $\ge \Delta - (\tau{+}\rho)\widehat\mu$ up};
\end{tikzpicture}
\caption{The chain argument of \Cref{lem:locality}. A wrong height offset $\Delta$ gives every
window point a partner $\Delta$ higher and within $\gamma\widehat\mu$ longitudinally. Climbing from $x$,
the chain can only stop by leaving through the window's edge, which takes
$\gtrsim \lambda/(\gamma\widehat\mu)$ steps. The climb would then exceed the spread cap $\beta\lambda$ unless
$\Delta = O(\widehat\mu)$.}
\label{fig:chain}
\end{figure}

\paragraph{Signs.}
The sign $\varepsilon_x$ is not known to the algorithm, and \Cref{rem:per-node-signs} shows it
cannot be shared across windows. Call $x$ \emph{free} if $A_+(x)$ and $A_-(x)$ are both nonempty
and differ. Otherwise $x$ has the single \emph{option set} $A(x) := A_+(x) \cup A_-(x)$. For sets
$P \subseteq Y$ of targets of $x$ and $Q \subseteq Y$ of targets of $z$, let
\[
\operatorname{dev}(x, P; z, Q) \;:=\; \max_{y \in P,\ y' \in Q}\ \bigl|d_X(x, z) - d_Y(y, y')\bigr|,
\]
and set $T := (2\Lambda + 1)\,\widehat\mu$. For a threshold $\theta \ge 0$, a sign pair $(s, t)$ is
\emph{$\theta$-forbidden} for $(x, z)$ if $\operatorname{dev}(x, A_s(x); z, A_t(z)) > \theta$. Each
prohibition is a clause on at most two free variables, so for each $\theta$ the prohibitions form a
2-SAT instance. A violated prohibition between two points that
are not free is the empty clause, and makes the instance unsatisfiable. For a point $x$ that is not free, we write
$A_+(x) = A_-(x) := A(x)$ in all that follows; likewise for $B_\pm(y)$.

\begin{lemma}[The true signs are consistent]
\label{lem:truth}
The signs $x \mapsto \varepsilon_x$ of \Cref{lem:window} violate no $T$-prohibition.
\end{lemma}

\begin{proof}
Since $A_{\varepsilon_x}(x)$ contains a partner of $x$, it is nonempty; so if $x$ is not free, then
$A_{-\varepsilon_x}(x)$ is empty or equal to it, and the option set of $x$ is $A_{\varepsilon_x}(x)$. By \Cref{lem:locality}, every target $y$ of $x$ under $\varepsilon_x$ has a partner
$j$ with $\|j - x\| \le \Lambda\widehat\mu$, and every target $y'$ of $z$ under $\varepsilon_z$ has a
partner $j'$ with $\|j' - z\| \le \Lambda\widehat\mu$. Then
\[
|d_X(x,z) - d_Y(y,y')| \;\le\; |d_X(x,z) - d_X(j,j')| + |d_X(j,j') - d_Y(y,y')|
\;\le\; 2\Lambda\widehat\mu + \mu \;\le\; T. \qedhere
\]
\end{proof}

\section{The algorithm, the main theorem and the constant}
\label{sec:main}

Run the construction in both directions. $X$-side windows give targets $A_s(x) \subseteq Y$, and
$Y$-side windows, read in $Y$'s frame, give sources $B_t(y) \subseteq X$. For each threshold $\theta$, one 2-SAT instance carries
three kinds of $\theta$-prohibitions:
\begin{itemize}[leftmargin=*]
\item \emph{$X$--$X$}: as above;
\item \emph{$Y$--$Y$}: the mirror image;
\item \emph{mixed}: $(s, t)$ is $\theta$-forbidden for $(x, y) \in X \times Y$ if
  $\max_{y_1 \in A_s(x),\ x_1 \in B_t(y)} |d_X(x, x_1) - d_Y(y_1, y)| > \theta$.
\end{itemize}
The algorithm uses the \emph{least} threshold $\theta^*$ whose instance is satisfiable. It exists when
no option set is empty, since a large $\theta$ forbids nothing, and it is $0$ or one of the $O(n^2)$
deviation values, found by binary search with $O(\log n)$ 2-SAT solves. A fixed threshold $T$ would
give the same guarantee but a worse algorithm (\Cref{rem:lower}). Given a solution at $\theta^*$, the algorithm picks \emph{any} $f(x) \in A_{s_x}(x)$ and \emph{any}
$f'(y) \in B_{t_y}(y)$, and forms $R := \{(x, f(x)) : x \in X\} \cup \{(f'(y), y) : y \in Y\}$.

\begin{algorithm}[ht]
\caption{Gromov--Hausdorff approximation in the plane}
\label{alg:main}
\begin{algorithmic}[1]
\Require finite $X, Y \subset \R^2$; parameters $\nu_0, a, \beta$ and $\varepsilon$ (\Cref{tab:constants})
\Ensure a correspondence $R$ with $\tfrac12\Dist(R) \le K\,\dGH(X,Y)$ \Comment{\Cref{thm:main}}
\If{$|X| = 1$ or $|Y| = 1$} \Return $X \times Y$ \EndIf
\State $R_{\mathrm{best}} \gets X \times Y$; fix a diameter pair $(x_\ell, x_r)$ of $X$ and read $X$ in its frame
\For{each ordered pair $y \ne y'$ of $Y$} \Comment{the frame of $Y$, \Cref{lem:frame}}
  \For{each guess $\widehat\mu = \nu_0 D\,(1 + \varepsilon)^{-k}$, $k = 0, 1, \ldots$, through the first $\widehat\mu \le \delta_{\min}/(2\Lambda + 2)$}
    \State compute windows, gates, and the sets $A_\pm(x) \subseteq Y$, $B_\pm(y) \subseteq X$
      \Comment{\Cref{def:window,def:admissible}}
    \If{some option set is empty} \textbf{continue} \EndIf
    \State find the least $\theta^*$ whose 2-SAT instance of $X$--$X$, $Y$--$Y$ and mixed
      $\theta^*$-prohibitions is satisfiable, and a solution $(s, t)$ of it
    \State pick any $f(x) \in A_{s_x}(x)$ and $f'(y) \in B_{t_y}(y)$;
      $R \gets \{(x, f(x))\} \cup \{(f'(y), y)\}$
    \If{$\Dist(R) < \Dist(R_{\mathrm{best}})$} $R_{\mathrm{best}} \gets R$ \EndIf
  \EndFor
\EndFor
\State \Return $R_{\mathrm{best}}$
\end{algorithmic}
\end{algorithm}

\begin{theorem}[Gromov--Hausdorff distance in the plane]
\label{thm:main}
Let the parameters satisfy
\[
\text{(R1)}\ \ a > \gamma, \qquad \text{(R2)}\ \ \beta a \ge 4C_1, \qquad \text{(R3)}\ \ 2\Lambda + 1 \ge 1 + 1/\nu_0,
\]
and let $\varepsilon \in (0, 1]$. For finite $X, Y \subset \R^2$, \Cref{alg:main} returns a correspondence $R$ with
\[
\dGH(X, Y) \;\le\; \tfrac12\Dist(R) \;\le\; K\,\dGH(X, Y),
\qquad K := (1 + \varepsilon)\,K_0, \quad K_0 := 2\Lambda + 1,
\]
using $O(\varepsilon^{-1} n^6 \log(2 + D/\delta_{\min}))$ arithmetic operations when $|X|, |Y| \ge 2$. Here
$n = |X| + |Y|$, $D = \diam X$, and $\delta_{\min}$ is the least positive value of
$|d_X(a,b) - d_Y(a',b')|$ over all $a, b \in X$ and $a', b' \in Y$, not necessarily distinct. With the
values of \Cref{tab:constants}, $K_0 < 356$.
\end{theorem}

\begin{proof}
\emph{Output.} Every candidate $R$ is a correspondence: it meets every $x$ through $f$ and every $y$
through $f'$. So is $X \times Y$. Hence $\dGH \le \tfrac12\Dist(R_{\mathrm{best}})$.

\emph{Degenerate cases.} If $|X| = 1$ or $|Y| = 1$, then $X \times Y$ is the only correspondence.
If $\mu > \nu_0 D$, then $X \times Y$ has $\Dist \le D + \mu \le (1 + 1/\nu_0)\,\mu \le K_0\mu \le K\mu$ by (R3).

\emph{A good frame and guess.} Otherwise $\mu \le \nu_0 D$, and the analysis applies; (R1) and (R2) are the hypotheses of \Cref{lem:locality,lem:window}. Take the frame
formed by partners of $x_\ell, x_r$; it is tried. Every $|d_X(a,b) - d_Y(a',b')|$ is either $0$ or
at least $\delta_{\min}$, and $\mu$ is one of these values. If $\mu > 0$, then
$\delta_{\min}/(2\Lambda + 2) < \delta_{\min} \le \mu \le \nu_0 D$, so the ladder reaches a guess
$\widehat\mu \in [\mu, (1 + \varepsilon)\mu]$ before it stops. If $\mu = 0$, take the last guess,
$\widehat\mu \le \delta_{\min}/(2\Lambda + 2)$, where $T = (2\Lambda + 1)\widehat\mu < \delta_{\min}$. The lemmas
use only $\mu \le \nu_0 D$ and $\mu \le \widehat\mu$, never $\widehat\mu \le (1 + \varepsilon)\mu$, so at that frame and guess
\Cref{lem:frame,lem:window,lem:locality,lem:truth} hold on both sides. No option set is empty,
since each contains a partner.

\emph{The true signs solve the $T$-instance.} $X$--$X$ and $Y$--$Y$ $T$-prohibitions are respected
by \Cref{lem:truth}. For a mixed prohibition, let $y_1 \in A_{\varepsilon_x}(x)$ and
$x_1 \in B_{\varepsilon_y}(y)$ be arbitrary. \Cref{lem:locality} on each side gives partners
$(j, y_1)$ and $(x_1, y'')$ in $R^*$ with $\|j - x\| \le \Lambda\widehat\mu$ and
$\|y'' - y\| \le \Lambda\widehat\mu$, so
\[
|d_X(x, x_1) - d_Y(y_1, y)| \;\le\; |d_X(x, x_1) - d_X(j, x_1)| + |d_X(j, x_1) - d_Y(y_1, y'')|
+ d_Y(y'', y) \;\le\; 2\Lambda\widehat\mu + \mu \;\le\; T.
\]
So the $T$-instance is satisfiable, hence $\theta^* \le T$, and the algorithm produces a candidate $R$.

\emph{Every solution is good.} Each pair of elements of $R$ is an $X$--$X$, $Y$--$Y$ or mixed pair,
and its distortion is bounded by the corresponding $\theta^*$-prohibition. So
$\Dist(R) \le \theta^* \le T$ for every solution at $\theta^*$ and every choice of $f$ and $f'$. If $\mu > 0$ this is
$T \le (2\Lambda + 1)(1 + \varepsilon)\,\mu = K\mu$. If $\mu = 0$, then $\Dist(R) < \delta_{\min}$ is a maximum
of values that are $0$ or at least $\delta_{\min}$, so $\Dist(R) = 0$. In both cases
$\tfrac12\Dist(R_{\mathrm{best}}) \le \tfrac12 K\mu = K\,\dGH$.

\emph{Time.} There are $O(n^2)$ frames and $O(\varepsilon^{-1}\log(2 + D/\delta_{\min}))$ guesses, since the
ladder runs from $\nu_0 D$ to $\delta_{\min}/(2\Lambda + 2)$ in steps of ratio $1 + \varepsilon$ and $\ln(1 + \varepsilon) \ge \varepsilon/2$. For each, windows
and gates cost $O(n^2)$. The sets $A_\pm$ and $B_\pm$ cost $O(n^3\log n)$: for each $x$, each
$y \in G(x)$ and each $z \in W(x)$, a binary search in $G(z)$ sorted by height. The prohibitions
cost $O(n^4)$, as each pair takes the extremes of $d_Y$ over two option sets; the deviations are
computed once per guess. Each 2-SAT solve is linear in the $O(n^2)$ clauses, the binary search for
$\theta^*$ makes $O(\log n)$ of them, and $\Dist(R)$ costs $O(n^2)$.
\end{proof}

\begin{table}[ht]
\centering
\small
\begin{tabular}{@{}llll@{}}
\toprule
& Constant & Definition & Value \\
\midrule
parameters & $\nu_0$ & threshold: $\mu \le \nu_0 D$ & $1/200$ \\
& $a$ & least window half-length $a\widehat\mu$ & $25/2$ \\
& $\beta$ & spread cap $s(W(x)) \le \beta\lambda(x)$ & $3$ \\
& $\varepsilon$ & ladder ratio $1 + \varepsilon$ & free \\
\midrule
frame, \Cref{lem:frame} & $c_0$ & $(6 + \nu_0)/(2(1 - \nu_0))$ & $< 3.018$ \\
& $c_r$ & $(4 + 2\nu_0)/(1 - \nu_0)$ & $< 4.031$ \\
& $g$, $\gamma$ & $c_0$;\ \ $g + c_0$ & $< 3.018$;\ \ $< 6.036$ \\
windows, \Cref{lem:window} & $e_s$, $E_s$ & \eqref{eq:window-consts} & $< 4.156$;\ \ $< 316.2$ \\
& $e_g$, $C_1$ & \eqref{eq:window-consts} & $< 4.704$;\ \ $< 9.292$ \\
& $\rho$, $\tau$ & $\max(2\sqrt{E_s}, 3C_1)$;\ \ $2\rho$ & $< 35.57$;\ \ $< 71.13$ \\
locality, \Cref{lem:locality} & $H$ & $4\rho + \beta a\gamma/(a - \gamma)$ & $< 177.27$ \\
& $\Lambda$ & $\sqrt{\gamma^2 + H^2}$ & $< 177.37$ \\
\Cref{thm:main} & $K_0$, $K$ & $2\Lambda + 1$;\ \ $(1 + \varepsilon)K_0$ & $< 356$;\ \ $< 360$ at $\varepsilon = 1/100$ \\
\bottomrule
\end{tabular}
\caption{The constants of \Cref{sec:frame,sec:locality,sec:main}. The requirements of
\Cref{thm:main} are (R1) $a > \gamma$, (R2) $\beta a \ge 4C_1$ and (R3) $2\Lambda + 1 \ge 1 + 1/\nu_0$;
with the values shown, $\gamma < 6.04 < 12.5$, $4C_1 < 37.17 \le 37.5$ and $K_0 > 201$.}
\label{tab:constants}
\end{table}

\begin{remark}[Constants and model]
\label{rem:constants}
The values in \Cref{tab:constants} were chosen to make $K_0$ small subject to (R1)--(R3); the structure
of the proof, not the choice, fixes the order of magnitude. About a factor $10$ separates $K_0$ from
$\rho$: $\tau = 2\rho$ for admissibility, $H \approx 5\rho$ in the chain, and the two endpoints of a
pair in \Cref{lem:truth}. The ladder starts at $\nu_0 D$ because larger guesses are never needed, and
its ratio trades the factor $1 + \varepsilon$ in $K$ against the factor $\varepsilon^{-1}$ in the
running time. The bound $O(\varepsilon^{-1} n^6\log(2 + D/\delta_{\min}))$ is polynomial in the bit size
of the input. For a strongly polynomial bound, replace the ladder by the $O(n^4)$ guesses
$|d_X(a,b) - d_Y(a',b')|$ with positive value, one of which is $\mu$ when $\mu > 0$, together with one
guess $\widehat\mu \le \delta_{\min}/(2\Lambda + 2)$. That gives $K = K_0 < 356$ in $O(n^{10})$
operations.
\end{remark}

\begin{proposition}[A lower bound for \Cref{alg:main}]
\label{prop:lower}
With the values of \Cref{tab:constants} and $\kappa := 2.35$, there are finite $X, Y \subset \R^2$ on which \Cref{alg:main}, with suitable
admissible choices of $f$ and $f'$, returns a correspondence $R$ with
\[
\tfrac12\Dist(R) \;\ge\; (\tau\kappa/g)\,\dGH(X, Y), \qquad \tau\kappa/g > 55 .
\]
\end{proposition}

\begin{proof}
Let $0 < \eta \le D/2000$ and $\delta := \tau\kappa\eta/g$, so that $\delta < 0.03D$. Let
$X = \{A, B, C, P_1, P_2\}$ with $A = (0,0)$, $B = (D,0)$,
$C = (D/2,\, 0.78D)$, $P_1 = (0.3D,\, 0)$ and $P_2 = (0.3D,\, \delta)$. Let $Y$ consist of $A' = A$,
$P_1' = P_1$, $P_2' = P_2$, $B' = (D + \eta, 0)$, and the point $C'$ above the axis with
$|A'C'| = |AC| + \eta$ and $|B'C'| = |BC| - \eta$.

\emph{Step 0: $2\dGH \le \eta$.} The identity $z \mapsto z'$ changes the three distances among $A$,
$B$, $C$ by exactly $\eta$, the distances from $P_1, P_2$ to $B$ by at most $\eta$, and the
distances from $P_1, P_2$ to $C$ by about $0.3\eta$. With $p := |AC| = |BC| \approx 0.9265D$, the
formula in the proof of \Cref{lem:frame} gives
\[
\pi(C') - \pi(C) \;=\; \Bigl(\frac{2p}{D + \eta} + \frac12\Bigr)\eta \;\ge\; \kappa\eta .
\]

\emph{Step 1: only three frames.} The diameter pair of $X$ is $\{A, B\}$, and the algorithm orders it
from $X$ alone. Suppose it takes $x_\ell = A$. The order matters: read from $B$, the shift of $C'$ is
only $(2p/(D + \eta) - \tfrac12)\eta$. If the algorithm takes $x_\ell = B$, use instead the mirror
image of $Y$: $B' = B$, $A' = (-\eta, 0)$, $|B'C'| = |BC| + \eta$ and $|A'C'| = |AC| - \eta$. The
argument below then applies with $A$ and $B$ exchanged, $0.3D$ replaced by $0.7D$, and $2.5\delta$ by
$1.1\delta$. Guesses are at most $\nu_0 D$, so gates have radius $g\widehat\mu \le g\nu_0 D < 0.016D$, and a
candidate needs every gate to be nonempty. The points of $X$ have $\pi \in \{0,\ 0.3D,\ 0.5D,\ D\}$. In
the frame of $(y, y')$ every $\pi(y'')$ is at most $|y - y''|$, and $|A'B'|$ is the only distance
of $Y$ above $0.94D$. So the gate of $B$ forces $y \in \{A', B'\}$ and an axis within $15^\circ$ of
the line $A'B'$, that is, $y' \ne C'$. For $y = B'$ the points of $Y$ have $\pi$ near $0$,
$0.5D$, $0.7D$ and $D$, none within $0.016D$ of $0.3D$, so the gate of $P_1$ is empty. There remain
the frames of $(A', B')$ and $(A', P_1')$, which coincide, and of $(A', P_2')$.

\emph{Step 2: every candidate has $\tau\widehat\mu \ge \delta$.} In
both frames $C'$ is the only point of $Y$ with $\pi$ within $0.1D$ of $\pi(C) = 0.5D$. In the frame
of $(A', B')$ the gate of $C$ is nonempty only if $g\widehat\mu \ge \kappa\eta$, so
$\tau\widehat\mu \ge \tau\kappa\eta/g = \delta$. The frame of $(A', P_2')$ is rotated by $\vartheta$ with
$\tan\vartheta = \delta/(0.3D)$, so $\pi(C') \ge 0.5D + 2.5\delta$. The gate of $C$ then needs
$g\widehat\mu \ge 2.5\delta$, and again $\tau\widehat\mu \ge \delta$.

\emph{Step 3: the pair merges.} Consider a candidate. The points $P_1, P_2$ have the same $\pi$, and
$\lambda \le \max(a\widehat\mu, \delta/\beta) < 0.1D$, so $W(P_1) = W(P_2) = \{P_1, P_2\}$ and
$G(P_1) = G(P_2) = \{P_1', P_2'\}$: in the frame of $(A', B')$ both images have $\pi = 0.3D$, in the
rotated frame they have $\pi$ within $2\delta^2/D$ of $0.3D$, below $g\widehat\mu \ge 2.5\delta$, and all other
points of $Y$ are at least $0.15D$ away. The heights of $P_1, P_2$ are $0$ and $\delta$, and the heights of $P_1'$,
$P_2'$ differ by $\delta' \in (0.99\delta, \delta]$. For each of the two points, each sign and each
of the two targets, the offset of \Cref{def:admissible} predicts for $P_1$ and $P_2$ heights within
$\delta$ of $h(P_1')$ or $h(P_2')$; the eight cases are immediate. Since $\tau\widehat\mu \ge \delta$, this
gives $A_\pm(P_1) = A_\pm(P_2) = \{P_1', P_2'\}$. So neither point is free, and both have the option
set $\{P_1', P_2'\}$. Then
$\operatorname{dev}(P_1, \{P_1', P_2'\};\, P_2, \{P_1', P_2'\}) \ge |\delta - 0|$, so
$\theta^* \ge \delta$, and the choice $f(P_1) = f(P_2)$ gives $\Dist(R) \ge \delta$.

Every candidate can thus have $\Dist(R) \ge \delta$, and $\Dist(X \times Y) \ge D$. Since
$2\dGH \le \eta$, the ratio is at least $\delta/\eta = \tau\kappa/g$.
\end{proof}

\begin{remark}[How far from tight]
\label{rem:lower}
So the worst case of \Cref{alg:main} lies between $55$ and $K$. The loss in
\Cref{prop:lower} is the freedom to pick \emph{any} admissible target: a pair of points closer than
the tolerance $\tau\widehat\mu$ can be merged. Closing the gap needs a polynomial-time rule that picks,
for each point, a target within $O(\mu)$ of a partner. When $\mu \le \nu_0 D$, minimizing $\Dist(R)$
over all admissible choices would be exact, since partners are admissible, but the obvious search is
exponential.

The least threshold matters too. On the same family, the window $\{P_1, P_2\}$ is a vertical pair,
and exchanging the two images fits it exactly under the sign $-$. If the $X$-side takes the sign
$-$ and the $Y$-side the sign $+$, then $R$ contains $(P_1, P_2')$ and $(P_1, P_1')$, a pair of
distortion $\delta$, and no $T$-prohibition forbids this while $\delta \le T$. With a fixed threshold
$T$ and an arbitrary solution the ratio would therefore reach about
$K_0\kappa/g \approx 277$. The true signs have deviation $O(\eta)$ here, so the least threshold excludes this
solution.
\end{remark}

\section{Open questions}
\label{sec:discussion}

\paragraph{The bijective distance.}
For $|X| = |Y|$, the frame, windows and signs apply to an optimal bijection unchanged, and
\Cref{lem:truth} holds for it. The bijective problem needs, in addition, a perfect matching between
$X$ and $Y$ that respects the chosen signs. We do not know that one always exists, and leave the
planar bijective problem open; equivalently open is the least additive distortion of
\citet{HallPapadimitriou2005} in the plane. \Cref{prop:mult-r1} shows that it cannot be approximated
through $\dGH$ alone.

\paragraph{The approximation ratio.}
The factor $K_0 < 356$ is worst-case bookkeeping (\Cref{rem:constants}); for \Cref{alg:main}
itself the worst case lies between $55$ and $K$ (\Cref{prop:lower}). The best ratio
attainable in polynomial time is open. We do not address dimensions $d \ge 3$ here.

\section*{Acknowledgments}

The author thanks Pankaj Agarwal and Yusu Wang for conversations at ICERM in
May 2026, and Jeffrey Vitter, Carola Wenk, and Boris Aronov for
many discussions on the Gromov--Hausdorff distance in Euclidean space.

\end{document}